\documentclass[runningheads]{llncs}
\usepackage{graphicx}
\title{Arc-routing for winter road maintenance}
\date{}

\begin{document}
\title{Arc-routing for winter road maintenance}
%
%
\author{Jiří Fink\inst{1,2}\orcidID{L-7284-2015} \and
Martin Loebl\inst{3,4}\orcidID{0000-0001-7968-0376} \and
Petra Pelikánová\inst{3}}
\authorrunning{J. Fink et al.}
%
\institute{Department of Theoretical Computer Science and Mathematical Logic, Charles University, Czech Republic
\email{fink@ktiml.mff.cuni.cz}\and
This research is conducted within the project Network Optimization (17-10090Y) supported by Czech Science Foundation\and
Department of Applied Mathematics, Charles University, Czech Republic
\and
supported by the H2020-MSCA-RISE project CoSP- GA No. 823748 
\email{loebl@kam.mff.cuni.cz}}

%
\maketitle              
\begin{abstract}
The winter road maintenance arc-routing is recognised as a notoriously hard problem not only from the algorithmic point of view. This paper lays down foundations of theoretical understanding of our new winter road maintenance optimization for the Plzen region of the Czech Republic which has been implemented by the regional authorities since the winter of 2019-20.
Our approach is not, contrary to most of existing work, based on the integer and linear programming machinery. We concentrate on studying arc-routing on trees. This is practical since routes of single vehicles can be well represented by trees, and allows algorithms and complementary hardness results. We then extend the approach to the bounded tree width graphs. This leads to considering the planar graphs which well abstract the realistic road networks.
 We formalise important aspects of the winter road maintenance problem which were not formalised before, e.g., public complaints. 
 The number of complaints from public against the winter road maintenance  is a quantitative measure of the quality of the service which is focused on, e.g., in media or in election campaigns. A fear of 'complaints' is a fact every optimizer must deal with. Hence, a formal model of public complaints and its inclusion in the optimization is vital.
 Our formalisation of the winter road maintenance is robust in the sense that it relates to well-know extensively studied concepts of discrete mathematics like graph cutting and splitting of necklaces. 

\end{abstract}
\newpage
\section{Introduction}
\label{s.int}

Our involvement started by being asked 

\medskip

{\em Can you improve routing for winter road maintenance in the Czech Republic and specifically in the Plzeň region.}

\medskip


\medskip

We were asked to create new routing for vehicles of winter road maintenance while minimizing the total number of used vehicles. 
There were many additional conditions that needed to be satisfied, in particular, conditions given by the Czech legislation. 
A fixed plan for one whole winter season had to be created. 

Our plan (described in \cite{FLP}) has been implemented by the Plzen region authorities starting the winter of 2019-20.
\medskip\noindent

{\bf Towards a model.} In the design of tours for vehicles in winter road maintenance, one needs to cover the graph of the road network by subgraphs and then one needs to design routing for each of these subgraphs by one vehicle. Each edge of the graph has attributes given by the {\em length}, the {\em priority} and the type of maintenance of the corresponding road segment. 
Some vertices serve as depots. Each such vertex has defined types of material which it can store. 

The road network has a {\em service priority} defined by the legislation based on traffic volume which partitions the roads into classes. For instance in the Czech Republic, there are three such classes: Arterial roads through regions have the highest level of service priority (1). Priority (2) is assigned to bus routes and other important routs. Third priority of service is assigned to local roads. 
 Each class of roads is associated with maximum time of maintenance completion. For instance, in the Czech Republic the edges in the first priority level have to be cleaned by a vehicle every three hours, in the second level every six hours and in the third level every twelve hours. 
 
 Next important issue is the {\em length of the working shift}. For instance, the standard length of the working shift of a maintenance driver is eight hours in the Czech Republic. Moreover, the Czech legislation requires multiple safety breaks for drivers during the working shift. It is natural to expand the time for the safety breaks and for all other non-driving manipulations of a vehicle to two hours per shift; this reduces the total time of driving to six hours.
 
 For simplicity, we will assume in our model that the working shift lasts a fixed amount of time, e.g., six hours during which there are no breaks for the drivers and also the time to load the maintenance material is negligible. We translate the time requirement into the {\em upper bound of the length} of the vehicle route.
 
 Another imposed rule (without clear rationale) is that each road has to be maintained by the same car in both directions. In this work we assume that {\em there are no one-way roads}.
The results presented here are valid also for the more general case when one-way roads are present; the assumption does not change the algorithmic and complexity considerations for the problems discussed while simplifies the definitions and the arguments. Also, this assumption has been valid for the winter road maintenance of the Plzeň region road network we optimised. This is the reason why we base our models on \emph{undirected} graphs in the next sections. For each undirected graph $G$ we consider its \emph{symmetric orientation} $G^s$ where each edge is replaced by two arcs with opposite orientations. This representation enables to discuss maintaining of edges in both directions.

Finally, there is the {\em capacity} $c_m$ describing the maximum length of a route which can be maintained with only one loading of the material $m\in M$.  The capacity condition requires that during each spreading material $m$ on road-length $c_m$ the vehicle must pass its depot $d$ at least once. It is convenient to define $c_m$ as a fraction of the maximum route-length of one vehicle.

 \medskip\noindent

The goal is to assign for each edge a vehicle which will maintain the corresponding road while minimizing the number of used vehicles and the length of the roads traversed without maintenance (\emph{deadhead}\index{deadhead}). A critically important part of the considerations are {\em public complaints}.

\medskip

    {\bf Features of a winter road maintenance plan}

    \begin{itemize}
        \item 
        We construct a partition $P= \{P_1, \ldots, P_r\}$ of the set of arcs of $G^s$ into sets $P_1, \ldots, P_r$ and for each $i$ we assign vertex (depot) $d_i\in D$. We assume the type of maintenance $m$ constant in $P_i$. We also assume that the oppositely oriented edges belong to the same $P_i$.
        \item
        We construct, for each $i$, set $R_i$ so that $P_i\subset R_i$ and each arc of $P_i$ may be reached from $d_i$ by a directed closed walk of $R_i$. 
        \item 
        For each $i$, we design a route servicing the edges of $P_i$ by a single vehicle starting and terminating at $d_i$ and using only arcs of $R_i$. The schedule must meet 
        
        (1) the requirement of the maximum length of the route,
        
        (2) the requirements given by priorities $p(e), e\in P_i$,
        
        (3) the requirements given by capacities $c_m$. 
        
        \item In our actual computation for the Plzen region descibed in \cite{FLP}, the steps above are performed simultaneously.

    \end{itemize}

\subsection{State of the art}
Winter road maintenance is recognised as a notoriously hard problem (not only) from the algorithmic point of view. As far as we know, most of the literature in the algorithmic winter road maintenance concentrates in designing algorithms, which are typically based on Integer Linear Programming (LP), Constrain Programming (CP) and  local heuristics. The complexity of such algorithms is at least exponential. 

An overview of literature on the problem of winter road maintenance and its solutions is \cite{PLC1,PLC2,PLC3,PLC4}. An excellent recent overview illustrating main works on the General Routing Problem can be found in \cite{BLMV} where the authors design a new branch-and-cut algorithm for the capacitated general routing problem. In \cite{PLA}, the authors also consider road priorities and a precedence relation between roads of different priority. In \cite{KHS}, the authors aim at constructing the routes schedule minimising the maximum length of a route; the network may have one-way streets and is modelled as a mixed graph.

Kinable et.al. \cite{KHS} study a real-world snow plow routing problem (in the USA) and they compare three methods based on Integer Linear Programming (LP), Constraint Programming (CP) and a local heuristic. Ciancio et.al. \cite{CDV} applied Branch-price-and-cut method for the Mixed Capacitated General Routing Problem with Time Windows. Other heuristic algorithms can be found e.g. in \cite{BYL,GMH,GB}.

In \cite{FLP} we introduced an heuristic approach with a very competitive implementation and described the computational results for the plan of the winter road maintenance in the Plzen region. Our plan has been implemented by the regional authorities. 

  \subsection{Main contribution}
  (1) Based on our experience with practical winter road maintenance we introduce several new concepts, including public complaints. We concentrate on studying these concepts first on trees, then on bounded tree width graphs and planar graphs.
  \newline
  (2) We relate these concepts to extensive research in discrete mathematics.
  \newline
  (3) We design algorithms based on dynamic programming and prove matching hardness results in most cases.
  
Summarising, we introduce a realistic robust model of winter road maintenance which can be successfully studied by theoretical methods and admits competitive algorithms without adding unrealistic conditions for the actual road networks to be maintained. 

\section{Basic concepts}
\label{s.bc}

In practice we have given a road network which we represent by a graph. 
Vertices represent crossroads (and dead ends) and edges represent roads among them. Let $z\geq 1$ denote the number of priority classes of roads and let $M$ denote the set of types of maintenance, e.g., $M= \{chemical,inert,snow-plow\}$ in the Czech Republic. We associate several functions with $G$:
    	\begin{itemize}
    		\item $\alpha: E\rightarrow R^+$ gives to every edge a non-negative length,
    		\item $p: E\rightarrow \{1,\ldots, z\}$ priority level,
    		\item $m: E\rightarrow M$ type of maintenance.
    	\end{itemize}
    	Let $D\subset V$ be a set of depots. For $d\in D$ we denote by $m(d)\in M$ the stored material at depot $d$. 
The arc routing problem is to search for a cover by subgraphs which correspond to parts of the network maintained by single vehicles.
   These subgraphs maintained by one vehicle are called maintaining plans.   

    \begin{definition}[Maintaining plan]
    \label{def.pre}
 \emph{Maintaining plan} is a tuple $(G, P, d, \alpha, z, p)$ where 

(1) $G= (V,E)$ is a graph of a road network, 

(2) $P\subset E$ is the set of maintained edges,

(3) $d\in V$ is the {\em depot},

(4) $\alpha: E\rightarrow Z^+$ gives to every edge a non-negative integer length,
    		
(5) $p: E\rightarrow \{1,\ldots, z\}$ gives to each edge its priority level. 
\end{definition}

Having a maintaining plan, we can search for a route which services it. There are several {\em external parameters} influencing properties of a servicing route.

\begin{definition} [External parameters]
\label{def.exp}
\begin{enumerate}
    \item 
    maximum length of a servicing route denoted by $L= L(G)$,
    \item  
    function $f\!: E \rightarrow Z^+$ giving an upper bound on the frequency of traversing each edge $e$,
    \item 
    function $t\!: \{1,\dots,z\}\times Z^+ \rightarrow Q^+$ describing limits associated with priorities: the total length of each servicing route between $i$-th and $(i+1)$-th traversal of edge $e$ is at most 
    $t(p(e),i)L$.
    \item
    capacity $c\leq 1$ such that the servicing route must visit the depot within each sub-route of total length bigger than $cL$.
\end{enumerate}
\end{definition}
These parameters are self-explanatory with the exception of function $t$, namely why it depends on specific traversals of a given edge: this is a natural feature of the winter road maintenance since it is most important that the roads are clean when people leave their homes in the morning, and when they come back home in the afternoon.

\begin{definition}[Vehicle route]
For given maintaining plan $(G, P,d,\alpha,z,p)$ we define a $L,c,t,f$-vehicle route as a closed walk $w= (e_1, \dots, e_\ell)$ where each $e_i$ is an element of the symmetric orientation $G^s = (V, E^s)$ of $G$, each edge $e\in E^s$ appears at least once and at most $f(e)$ times in $w$ and (1) the requirement of the maximum length of the route, (2) the requirements given by priorities $p(e), e\in P$ and (3) the requirements given by the capacity $c$ are kept. 
\end{definition}

\begin{definition} [Maintaining plan routing problem]
 \label{def.mpr}
{\em Maintaining plan routing problem} is to decide, given a maintaining plan with $P=E$, if a vehicle route exists. We say that the problem is {\em unweighted} if the length of each edge is equal to one.
\end{definition}

We note that the condition $P=E$ is a natural simplification when considering single routes since the dead-heading is typically negligible.

\medskip

We first observe that the problem to decide if
a vehicle route exists is NP-complete even if $G$ is a star rooted at its vertex $d$ of degree one, $c= 1/2$ and $t$ is uniformly equal to 1. 
Such input tree admits a vehicle route if and only if the edges not incident with $d$ can be divided into two parts with equal sums of edge lengths. This problem is called {\em Partitioning} and is a basic NP-complete problem.

\medskip

In view of this observation it is natural to consider a restriction that {\em the edge-lengths are integers bounded by a fixed power of the size of the input graph}. A possible justification is that in practice the resources (time, maintaining material, cost) depend linearly on the length of the road segments. However, it is important to have in mind that there are situations, e.g., steep hills with heavy snow-fall, when this is not true and on the contrary it is realistic to assume that the resources depend non-linearly, even exponentially, on the lengths of road segments.

\medskip

{\bf Trees.}
In this paper we concentrate mostly on the 
Maintaining plan routing problem on trees.
Trees are graphs useful for a representation of a vehicle route in the winter road maintenance. Even if the set of edges maintained by a single vehicle is not a tree we can represent it as a tree obtained e.g. from the Depth first search (DFS) algorithm. The symmetric orientation of a tree is always an Eulerian graph and thus there exists a natural vehicle route if the only goal is to visit all arcs exactly once.
We show in next sections that the 
Maintaining plan routing problem is interesting and not easy even for trees. 

\subsubsection{Graph cutting.}
\label{sub.gcut}

The routing when the priority function is constant is closely related to the 'classic' graph theory concept of {\em graph cutting}.

\begin{definition}[Graph Cutting Problem]
Graph Cutting Problem is to find, for a given graph $G$ rooted in its vertex $r$ and set of numbers $t_1, \dots, t_k$, a cover of $E(G)$ by connected subgraphs $G_1, \dots, G_k$ rooted in $r$ of sizes $t_1, \dots, t_k$ ($\leq t_1, \dots, \leq t_k$ respectively).
\end{definition}

We will need and prove a negative result on the graph cutting when we consider the planar graphs. However, there is a very nice positive result:

\begin{theorem}[\cite{JRP,G,L}]
\label{thm.gl}
Given a k-edge-connected graph $G = (V,E)$, k edges $e_1,e_2, \ldots, e_k$ of $G$ and k positive integers $m_1, \ldots, m_k$ with the sum equal to $|E|$. There exists a partition $E = E_1\cup \cdots \cup E_k$ such that $e_i \in E_i$, $|E_i|= m_i$, and $G_i = (V(E_i),E_i)$ is connected for each $i\leq k$.
\end{theorem}

\medskip

The particular aspect of winter road maintenance introduced next are {\em public complaints}. This is an important issue of anybody in this business all around the world.
 The number of complaints from public against the winter road maintenance  is a quantitative measure of the quality of the service which is focused on, e.g., in media or in election campaigns. 

\subsection{Public complaints}

The experience is that residents make complaints to insufficient service if they think that they are treated in an unfair manner in particular if their neighbourhood is 'skipped' in the service.

We make a {\em rational assumption} that the number of public complaints can be deduced from the structure of a vehicle route and in particular from its {\em perceived unfairness}. We call this (number) the \emph{unfairness index} of a vehicle route $w$, and define it in the Definition \ref{def.neig} below. 

The additional structure used in this definition is {\em a collection of fixed cyclic orders of the neighbours of each vertex}. We observed the empiric existence and the importance of such orders in our practical work for the winter road maintenance. In fact, in situations when a consensual order between two edges sharing a vertex (representing two road-segments sharing a crossing) does not exist, the administrators in charge of the winter road maintenance insisted on using different maintaining cars to service these two roads in order to avoid complaints caused by one vehicle giving 'unjustified preference' to one of the two neighbourhoods.

\begin{definition}
\label{def.neig}
Let $G$ be a graph and let $d$ be its vertex (the depot) of degree 1. We further assume that we are given a fixed cyclic order $O(v)$ of the neighbours of each vertex $v$. Let $w= (e_1, \ldots, e_l)$ be a vehicle route. For $i< l$ let $e_i= (s_i, t_i)$, let $(w,i)^+$ denote the edge of $G$ incident with $t_i$ which follows $\{s_i, t_i\}$ in $O(t_i)$ and let $(w,i)^-$ denote the edge of $G$ incident with $s_i$ which precedes $\{s_i, t_i\}$ in $O(s_i)$.
\begin{itemize}
    \item If no orientation of $(w,i)^+$ belongs to $(e_1, \ldots, e_{i+1})$ then we say that edge $(w,i)^+$ has a {\em forward complaint}.
    \item If no orientation of $(w,i+1)^-$ belongs to $(e_1, \ldots, e_{i+1})$ then we say that edge $(w,i+1)^-$ has a {\em backward complaint}.
    \item
    The \emph{unfairness index} of the route $w$, denoted by $Uf(w)$, is the sum of the number of edges which have a forward complaint and the number of edges which have a backward complaint.
    \end{itemize}
    \end{definition}

Naturally we can introduce the {\em unfairness minimisation problem} to find a vehicle route $w$ with $Uf(w)$ as small as possible.
We show further that this innocent looking problem is related to the extensively studied \emph{necklace splitting problem}.

\section{Main results}
\label{s.mainr}

\subsection{Routing problem on trees with constant priority function}

 We recall that the parameters of winter maintaining of a network are
(1) upper bound $L$ for the total length of a vehicle route, (2) priority function $t$ and (3) capacity $c$. Additional assumption of this section is a constant priority function. We denote the value $t:= t(p(e),i)$. 

\medskip\noindent


We already noticed that the weighted problem is NP-complete even for stars. Hence, let us consider the unweighted problem. It is equally straightforward to observe that deciding admissibility for the unweighted problem is NP-complete for subdivided stars when the capacity $c$ may depend on the input tree (reduction to the 3-partitioning problem).

\medskip

The case of the fixed capacity $c$ already admits a polynomial algorithm based on the dynamic programming.

\begin{theorem}\label{thm:sufficient_cond}
There is a polynomial algorithm for finding a solution of the unweighted maintaining plan routing problem restricted to maintaining plans $(T,d,p)$ with $T$ tree, the priority function $t$ constant and the capacity $c=1/c'$ a constant fraction not depending on the input tree. Also, the unfairness minimisation problem admits a polynomial algorithm.
\end{theorem}

\subsection{Public complaints and necklace splitting}
Noga Allon \cite{splitting_necklaces} studied in 1987 an interesting problem in combinatorics which may be interpreted as the problem how to divide a stolen necklace fairly between the thieves.

\begin{definition}[k-splitting] 
    Let $N$ be an open necklace, i.e., a path consisting of $k\cdot n$ vertices-beads, chosen from $s$ different colors. There are $k\cdot a_i$ beads of color $i$, $1\leq i \leq s$. A \emph{$k$-splitting}\index{k-splitting} of the necklace is a partition of the necklace into $k$ parts, each consisting of a finite number of non-overlapping intervals of beads whose union contains precisely $a_i$ beads of color $i$, $1\leq i \leq s.$ The \emph{size}\index{k-splitting! size} of the $k$-splitting is the number of cuts that forms the intervals of the splitting.
\end{definition}

\begin{definition}[Necklace splitting problem] 
    Let $N$ be a necklace. \emph{Necklace splitting problem}\index{necklace splitting problem} is to find for given number $k$ a k-splitting of necklace $N$ of minimal size.
\end{definition}

If the beads of each color appear contiguously, then at least $k-1$ cuts  between the beads of each color are necessary and hence the number $(k-1)\cdot s$ of cuts is a lower bound. The following theorem says that this is sufficient for all $k$-splittings.

\begin{theorem}[Noga Alon]
\label{thm:alon}
    Every necklace with $ka_i$ beads of color $i$, $1\leq i \leq s$, has a $k$-splitting of size at most $(k-1)\cdot s.$ 
\end{theorem}

This theorem has only topological non-constructive proofs so far; Alon's proof uses a transformation of the discrete problem to a continuous coloring of the unit interval.
    
\subsubsection*{Complexity of necklace splitting.}

The algorithmic complexity of the necklace splitting has been intensively studied. 
First, the problem to
\emph{ determine the algorithmic complexity of feasible splitting with the smallest number of cuts} was proven to be NP-complete even for 2-splitting ($k= 2$) and two beads of each color by Bonsman, Eppig and Hochstättler \cite{bonsma}. Alternative proof was made by Meunier \cite{meunier}.

However, more attention has been given to another problem. Since the known proofs of the existence of the splitting of size $(k-1)s$ are not constructive, the consequent research has been directed towards constructively finding the splitting. The following question had been open for a long time: 

\medskip

\emph{Can one find efficiently the splitting guaranteed by Theorem \ref{thm:alon}?} 

\medskip

This was finally answered negatively in 2019 by Filos-Ratsikas and Goldberg \cite{FG}. To explain this result we introduce the problem \emph{LEAF} (see \cite{P,FG}).

\begin{definition}[LEAF problem]
    An instance of the problem called \emph{LEAF}\index{LEAF} consists of a graph $G$ of maximum degree $2$, whose $2^n$ vertices are represented by  $0,1$ sequences of length $n$; $G$ is given by a polynomial Turing machine that takes as input a vertex and outputs its neighbours. Moreover, the vertex $0$ has degree $1$. The goal is to output another vertex of degree $1$.
\end{definition}

We say that a problem is \emph{PPA-complete}\index{PPA-complete} if it is polynomial time equivalent to the LEAF problem. 
A cryptographic hardness of the PPA-complete problems is discussed e.g. in \cite{FG}.

\medskip

The result of Filos-Ratsikas and Goldberg is that finding necklace splitting guaranteed by Theorem \ref{thm:alon} is PPA-complete even for $k=2$.

\medskip

As discussed earlier, the {\em weighted} maintaining plans is relevant for the winter road maintenance when it is realistic to assume that the resources of vehicles (time, amount of the spreading material) depend non-linearly on the lengths of road segments.

\begin{theorem}\label{thm:splitting} 
    There exists a polynomial reduction of the necklace splitting problem to the  unfairness minimisation for maintaining plans routing on trees, even when the maintaining plan is a star with weights on edges.
\end{theorem}

Finally, we find the next questions appealing:

\begin{question}
\label{q.aaa}
Is there a good approximation algorithm for the weighted unfairness minimisation arc routing for trees?
\end{question}

\begin{question}
\label{q.aaaa}
Is there an analogue of Theorem \ref{thm:alon} for the unfairness minimisation for general trees and for planar graphs?
\end{question}

\subsection{Routing unweighted trees with bounded degrees}
\label{s.adm}
In this section we consider general priorities in the maintaining plans. All trees are unweighted. 
We construct a polynomial algorithm based on dynamic programming which can decide if a given maintaining plan $(T, d, z, p)$, $T$ tree {\em of bounded degree} admits a vehicle root where in addition {\em each arc is traversed at most a constant number of times}. We also show that both these additional assumptions are necessary.

\begin{theorem} \label{thm:fix_degree_traverse}
Fixed integers $F, \Delta$. There exists a polynomial time algorithm which for a tree $T = (V,E)$ rooted in $d$ with maximal degree at most $\Delta$, function $f: E \to N$ such that $f(e) \le F$ for all $e \in E$
and function $g: E\times \{1, \ldots, F\} \to N$ decides whether there exists a closed walk $w$ starting at $r$ satisfying
\begin{itemize}
\item Every edge $e$ of $T^s$ is traversed $f(e)$-times (at most $f(e)$-times respectively)  in both directions.
\item For every edge $e$ of $T^s$ and $y\leq f(e)$, there are at most $g(e,y)$ steps between $y-$th and $(y+1)-$st traverses of $e$, taken cyclically.
\end{itemize}
As a consequence, there is a polynomial algorithm  to decide if a $L,c,t,f-$vehicle route on $T$ exists.
\end{theorem}

Question: Is it necessary to fix $F$ and $\Delta$? We first show that the admissibility is hard for unbounded $f$ even if $G$ is a binary tree and $g$ depends only on the edge.

\begin{theorem}
\label{thm.np1}
It is NP-complete to decide whether a given binary tree $T = (V,E)$ rooted in $d$ and functions $f,g : E \to N$ there exists a closed walk $w$ starting at $r$ satisfying
(1) Every edge $e$ is traversed $f(e)$-times in both directions (2) For every edge $e$, there are at most $g(e)$ steps between two consecutive traverses of $e$ in both direction, taken cyclically.
The problem is NP-complete even if we restrict $f$ to be non-increasing on all paths from the depot.
\end{theorem}

Next theorem treats the case unbounded degrees.

\begin{theorem}
\label{thm.np2}
Fixed integer $F$. It is NP-complete to decide whether a given tree $T = (V,E)$ rooted in $d$, function $f: E \to N$ such that $f(e) \le F$ for every edge $e$ and function
$g: E\times\{1, \ldots, F\} \to N$
there exists a closed walk $w$ starting at $d$ satisfying
(1) Every edge $e$ is traversed $f(e)$-times in both directions.
(2) For every edge $e$ and $y\leq f(e)$, there are at most $g(e)$ steps between the $y-$th and $(y+1)-$st traverses of $e$ in both direction, taken cyclically.
\end{theorem}

\subsection{Routing unweighted graphs of bounded tree-width}
\label{s.adm1}
In this section, all graphs will be unweighted. 
A {\em tree decomposition} of a graph $G$ is a pair $(W, b)$ where $W$ is a tree and $b: V(W) \rightarrow 2^{V(G)}$ assigns a {\em bag} $b(v)$ to each vertex
$v$ of $W$ such that
\begin{itemize}
    \item every vertex is in some bag,
    \item every edge is a subset of some bag,
    \item every vertex of $G$ appears in a connected subtree of the decomposition.
\end{itemize}

The {\em width} of the tree decomposition is defined as the size of the largest bag, minus one. The {\em tree-width} of graph $G$ is the minimum width of a tree decomposition of $G$.

Let $G = (V,E)$ have a distinguished vertex, denoted by $d$. It is useful to simplify the decomposition. A tree decomposition $(W,b)$ is {\bf canonical} if
\begin{itemize}
    \item $T$ is rooted, and the root $r$ satisfies $d\in b(r)$. 
    \item Each leaf $u$ satisfies $|b(u)|= 1$.
    \item Each non-leaf vertex $u$ satisfies one of the following conditions:
     
     $u$ has exactly one son $u'$ and $b(u) = b(u') \cup \{v\}$ for some vertex $v \in V$.
     
     $u$ has exactly one son $u'$ and $b(u) = b(u') \setminus \{v\}$ for some vertex $v \in V$.
     
     $u$ has exactly two sons $u', u''$ and $b(u) = b(u')= b(u'')$.
     
\end{itemize}

It is straightforward to verify that every graph $G$ of tree-width at most $k$ has a canonical tree decomposition of width at most $k$, of polynomial size. The following theorem is again proved by a dynamic programming argument building on the proof of Theorem \ref{thm:fix_degree_traverse}.

\begin{theorem}
\label{thm.c}
Let $z, \Delta, F$ be integer constants and let $(G,d,z,p)$ be a maintaining plan where $G = (V,E)$ is a graph rooted in $d$ and with maximal degree at most $\Delta$, given along with its canonical tree decomposition $(W,b)$ of width $k-1$ and functions $f: E^s \to N$ such that $f(e) \le F$ for all $e \in E^s$ and $t: E^s\times \{1, \ldots, F\} \to N$. Then there is an algorithm  to decide if a $L,c,t,f-$vehicle route on $G$ exists of complexity at most $pol(|G|,|t|)\times (4kF|E|)^{4Fk\Delta}$.
\end{theorem}

\subsection{Case of more routes} \label{sec:more_routes}
In this section all graphs are unweighted.
The graph of road network has each edge maintained by one method. We recall that the set of the possible maintaining methods is denoted by $M$. We will assume that $M$ has a fixed size, e.g., $M= \{c,i,s\}$. It is natural to assume that each maintaining vehicle can snowplow and thus we can include deadheading in our model.  

\begin{definition}
\label{def.mm}
    Let $G$ be a graph of road network and $G^s$ its symmetric orientation.  Let $p: E\rightarrow \{1,\ldots, z\}$ be its priority function and $m: E\rightarrow M$ be its
    maintaining type function. Let $D$ be the set of the depots. We say that a tuple $(H, P, d, \alpha, z, p)$ where $H$ is a subgraph of $G$ and $d\in D \cap V(H)$ is a {\em maintaining plan of $G$ } if $m$ is constant on $H$ and $(H,P,d,\alpha,z,p)$ admits a  $L,c,t,f-$vehicle route.
    \end{definition}

\begin{definition}[Feasible and Optimal Solution]
    Feasible solution of a road network $G$ is a set $O$ of admissible plans of $G$ so that 
    the union of their $P-$sets covers $E(G)$. A feasible solution $O$ is {\em optimal} if 
    $|O|$ is as small as possible.
    \end{definition}

\subsubsection{Finding an optimal solution.}

First we note that finding an optimal solution is an NP-complete problem even for $G$ a tree, all edge-weights equal to 1 and $|D|= 1$.

Hence from now on let $o$ be a fixed integer and we consider the optimization problem 
{\bf Roadnet(o):} find out if there is a feasible solution of a road network consisting of at most $o$ admissible plans.
We arrive at a result analogous to Theorem \ref{thm.c} by further refining the dynamic optimization argument of its proof.

\begin{theorem}
\label{thm.c1}
Let $z, \Delta, F$ be integer constants and let $(G,D,z,p,m)$ be a road network where $G = (V,E)$ is a graph with maximal degree at most $\Delta$ and $D\subset V$, given along with its canonical tree decomposition $(W,b)$ of width $k-1$ and functions $f: E^s \to N$ such that $f(e) \le F$ for all $e \in E^s$ and $t: E^s\times \{1, \ldots, F\} \to N$. Then there is an algorithm for  {\bf Roadnet(o)} of complexity at most $pol(|G|,|t|)\times (4koF|E|)^{4Fko\Delta}$.
\end{theorem}

\subsection{Solving routing in planar networks}

In this section, we 
consider the class of the planar graphs which realistically model most of road networks.
We start with a hardness result on planar graph cutting.

\begin{theorem}
\label{thm.npcut}
The following planar graph cutting problem is NP-complete: given a planar graph $G$, its vertex $d$ and numbers $t_1, t_2$, decide if there are two connected subgraphs $G_1, G_2$ containing $d$ so that $|E(G_i)|= t_i, i= 1,2$ and 
$E= E(G_1)\cup E(G_2)$.
\end{theorem}
\begin{proof}
We show a reduction from the Steiner tree problem for planar graphs which is a well known NP-complete problem.

{\em Steiner tree problem}: given a graph and a set $T$ of its vertices called terminals, find a connected subgraph that includes all the terminals and has the minimum possible number of edges.

The reduction goes as follows: let $G$ be a planar graph and let $T$ be a set of its vertices. We take one of the vertices of $T$ and call it $d$. Next, we attach to each vertex of $T\setminus d$ a path of $|E|$ edges. We let 
$t_1= (|T|-1)|E|+x$ and $t_2= |E|$. Obviously, $G$ has a Steiner tree of size at most $x$ if and only if a feasible graph cutting exists.

\end{proof}

Taking into account the proof of Theorem \ref{thm.npcut}, we get the following hardness result:

\begin{theorem}
\label{thm.hhh}
The maintaining plan routing problem for
the planar graphs is NP-complete even when $c= 1/2$ and all edge-weights are equal to one.
\end{theorem}


By Theorem \ref{thm.c1} (case of $o=1$),  the maintaining plan routing problem for the planar graphs with bounded degrees can be solved in $2^{O(\sqrt n log  n)}n^{O(1)}$ since every planar graph of $n$ vertices has tree width at most $\sqrt n$. 

We conjecture that assuming the exponential time hypothesis, there is no algorithm of complexity $2^{o(\sqrt n)}n^{O(1)}$.

Most of the realistic medium size road networks are planar bounded degree, with at most  ten thousand edges (road-segments) and around one hundred  of the maintaining vehicles. 
This leads to studying planar road networks with $n$ vertices and $O(\sqrt n)$ maintenance cars. 

We do not know if this problem admits a sub-exponential algorithm.
However, many realistic road networks contain small cuts and their tree-width is small. For such networks, Theorem \ref{thm.c1} implies a sub-exponential algorithm.

\section{The proofs}

\subsection{Proof of Theorem \ref{thm:sufficient_cond}}
\label{subt1}

We start by showing a polynomial algorithm for the Tree Cutting Problem based on the dynamic programming. 

\begin{theorem}
\label{thm.cuttk}
Let $k$ be any fixed number. There is a polynomial algorithm to solve the Tree Cutting Problem.
\end{theorem}

\begin{proof}
We are given a tree $T$ rooted in a vertex $r$ and integers $s_1,\dots,s_k$. Let us denote by $n$ the number of vertices of $T$. 
We need to cover the graph $T$ with trees of given sizes. To do that we will construct a cover for all possible sizes of trees. This set of covers is denoted by $F(v)$.
Formally, we proceed in two steps. 

First, define for each $v\in V(T)$ a set $F'(v)$; the elements of $F'(v)$ are all $k$-tuples of trees $(T_1, \dots, T_k)$ with root $v$ such that:
$$B(v) = \bigcup_{i = 1}^k T_i.$$
Each $k$-tuple in $F'(v)$ is a cover of the branch $B(v)$ by trees $T_1, \dots, T_k$. 
The {\em size} of $(T_1, \dots, T_k)$ as the vector $(t_1,\dots, t_k)$ where $t_j = |E(T_j)|$, $1\leq j \leq k$. 
We define an equivalence on $F'(v)$: 
    $(T_1,\dots, T_k)$ and $(T'_1,\dots, T'_k)$ are equivalent if $\forall i \in \{1, \dots, k\}\!: t_i = t'_i.$ 

Secondly, for each vertex $v$ we let $F(v)$ to be the set of all representatives of this equivalence. 
We denote by $S(v)$ the set of the sizes of the elements in $F(v).$
We note that $|S(v)| \leq n^k$ because each tree has at most $n$ edges and there are $k$ trees in every $k$-tuple.

We will construct $F(v)$ for all $v \in T$ recursively.

\begin{itemize}
    \item Let $v$ be a leaf: $F(v) := \{(\emptyset, \dots,\emptyset)\}.$
    \item Let $v$ be a parent of vertices $v_1, \dots, v_m$ and we assume for all $i = 1,\dots, m$ $F(v_i)$ are determined.
    We will define for $v$ and $1\leq i\leq m$ the set $F_{v_i}(v)$ of $k$-tuples representing covers of the subtrees induced by vertex $v$ and branches rooted in its children $v_1, \dots, v_i$. We construct $F_{v_i}(v)$ by the {\em additive step} described below.
    \item
    We let $F(v) := F_{v_m}(v).$
\end{itemize} 
\medskip

Now we describe the additive step used for the construction of $F(v)$. 

\medskip
\noindent {\bf Additive step:} Given a k-tuple $(T_1,\dots, T_k)$, we construct a k-tuple $(T'_1, \dots, T'_k)$ by addition of the edge $\{v_i, v\}$ in every possible way. There are two cases dependent on the size of the cover.\\[-7mm]
            \begin{enumerate}
            \item Case $t_i = 0$ for each $1 \leq i \leq k:$\\
                 we add the edge for every nonempty subset of indices\\ 
                 $$\forall I \subset \{1, \dots, k\}, |I| \geq 1$$
                            $$j \in I: T'_j = T_j \cup \{v, v_i\}$$
                            $$j \notin I: T'_j = \emptyset$$
            \item Case $\exists i \in \{1,\dots k\}: t_i \neq 0$\\
                Let $J  = \{i\;|\, t_i = 0\},$ 
                we add the edge for all nonepmty trees and
                we add the edge for every subset of empty trees 
                            $$\forall I \subset J$$
                            $$j \notin J: T'_j = T_j \cup \{v, v_i\}$$
                            $$j \in I: T'_j = T_j \cup \{v, v_i\}$$
                            $$j \in J\setminus I: T'_j = \emptyset$$
                            
            \end{enumerate}
This finishes the description of the Additive step.

\begin{itemize}
        \item Construction of ${F_{v_1}(v)\!:}$\par
            We start with $F_{v_1}(v) = \emptyset.$
            For each $(T_1, \dots, T_k) \in F(v_1)$ we construct by the additive step the set of $k$-tuples $(T'_1, \dots, T'_k)$ which we add into $F_{v_1}(v).$
        
        \item Construction of ${F_{v_{i+1}}(v)\!:}$\par
            In this case we proceed in two steps. 
            
            First we construct $F'(v_{i+1})$ by  additive steps applied to $F({v_{i+1}})$.
            Specifically, we start with $F'(v_{i+1}) = \emptyset.$
            For each $(T_1, \dots, T_k) \in F(v_{i+1})$ we construct by the additive step a set of $k$-tuples $(T'_1, \dots, T'_k)$ which we add into $F'(v_{i+1})$.
            
            Secondly, we merge $F_{v_i}(v)$ and $F'(v_{i+1})$ again in two steps as follows.
            
            First, for all $(T_{11}, \dots, T_{1k}) \in F_{v_i}(v)$ and all $(T_{21}, \dots, T_{2k}) \in F'(v_{i+1})$
            
            $$\forall i \in \{1, \dots, k\}\!: T'_i =   T_{1i} \cup T_{2i}$$
            and we add into $F_{v_{i+1}}(v)$ the created $k$-tuple $(T'_1, \dots, T'_k).$
            
            Finally, we clean the set $F_{v_{i+1}}(v)$ by keeping only the representatives of equivalence classes.
\end{itemize}

The described construction of the set $F(r)$ determines the set $S(r)$ of sizes. We have a solution of the tree cutting problem if and only if the $k$-tuple $(s_1,\dots, s_k)$ is in $S(r)$.  

By analyzing the above procedure we need at most two times $n^k \times n^k$ steps for adding of one edge. Hence the complexity of the algorithm is asymptotically $n^{2k+1}$ because there are at most $n$ edges.
 So there exists a polynomial algorithm for the tree cutting problem with fixed $k$.
\end{proof}

{\bf Proof of Theorem \ref{thm:sufficient_cond}.}
We distinguish several cases.


\begin{enumerate}
    \item \label{thm:const_prio_case1}$t \geq 1$, $c = 1\!:$
        \begin{itemize}
            \item priority condition and capacity condition always hold
            \item necessary and sufficient condition for existence of $L,c,t$-vehicle route  is: $|E(T)| \leq \frac{1}{2} L$
        \end{itemize}
    \item\label{thm:const_prio_case2} $t \geq 1$, $c < 1\!:$
        \begin{itemize}
            \item priority condition always holds
            \item clearly, if an arc belongs to a trip then its reverse belongs to the same trip. Trips are determined by subtrees rooted in $d$ and there is no advantage in going through an arc more than once in the same trip
            \item we want to construct subtrees $T_1, \dots, T_k$ rooted in $d$ such that $k := \left\lceil \frac{1}{c}\right\rceil$, 
            $$\forall i \in [k]: |T_i| \leq cL,$$ 
            $$T= T_1\cup \ldots T_k$$ and 
            $$\sum_{i\leq k}2|E(T_i)|\leq L.$$ 
            This is achieved by the algorithm for tree cutting problem described in the proof of Theorem \ref{thm.cuttk}. A solution is any collection of $k$ trees with sizes $t_i$ satisfying $t_i \leq cL$ for each $i\leq k$, and 
            $$\sum_{i\leq k}2t_i\leq L.$$ 
        \end{itemize}
    \item\label{thm:const_prio_case3} $t < 1$, $t \leq c\!:$
        \begin{itemize}
            \item capacity condition holds if priority condition is satisfied
            \item necessary and sufficient condition for existence of $L,c,t$-vehicle route is: $2|E(T)| \leq tL$.
        \end{itemize}
    \item\label{thm:const_prio_case4} $t < 1$, $c < t\!:$
        \begin{itemize}
            \item this case is equivalent to case 2. for $L':= tL$, $t'=1$ and capacity $c' := \frac{c}{t}$
        \end{itemize}
\end{enumerate}
Solution of cases 1 and 3 can be a DFS order of a tree if the necessary condition holds otherwise there is no solution. Case 4 is reduced to case 2. Case 2 admits a polynomial algorithm by Theorem \ref{thm.cuttk}.

\subsection{Proof of Theorem \ref{thm:splitting}}
\label{sub.t2}

\begin{proof}
    We have given an instance of the Necklace splitting problem. The necklace $N$ of length $nk$ has to be partitioned into $k$ parts each containing $a_i$ beads of color $i$, $1\leq i \leq s$. So the number of colors is $s$.
    
    We describe a construction of a network. The graph of the network will be a star with center $x$ rooted in its leaf $d$ and with $nk$ non-root leaves $u_1, \ldots, u_{nk}$. The cyclic order for $x$ is $(d,u_1, \ldots, u_{nk})$.
    The number of traverses of each arc $e$ will be bounded by $f(e)$ where $f(e)= 1$ if $e$ is not incident with $d$, and $f(e)= k$ otherwise.
    
     For each color $r$, $1\leq r\leq s$, we define number $M_r$ recursively: $M_1= 1 + n$ and  $M_{r+1}= 1 + n \sum_{q\leq r} M_q$.

    For each $i$, $1\leq i \leq nk$, if the $i$-th bead of the necklace has color $r$, than we let $\alpha(\{x,u_i\})= M_r$. We also let $\alpha(\{x,d\})= 1$. Finally we let 
   $L= 2k(1+\sum_{r=1}^s a_r M_r))$ and $c= \frac{1}{k}$.

    The capacity constant determines the length of each trip from the depot to be exactly $2(1+\sum_{r=1}^s a_r M_r)$. Multiplication by two means each edge is traversed in both directions. 
    
    A solution of the unfairness minimisation problem is a vehicle route $w$. The complaints at edges $\{x,u_i\}$ naturally determine the splits of the necklace. If $\{x,u_i\}$ has a forward complaint then we split the necklace between the $i-1$-th and $i$-th beads of the necklace. 
    If $\{x,u_i\}$ has a backward complaint then we split the necklace between the $i$-th and $i+1$-st beads of the necklace.
    
    Summarising, the sum of the complaints at vertices of the route minimising the unfairness index is equal to the minimum size of the necklace splitting.

\end{proof}

\subsubsection{Unfairness and sum packing problem}
As a follow-up to the proof of Theorem \ref{thm:splitting} we show a connection of the unfairness minimisation problem with the Sum Packing Problem of Erd{\"o}s \cite{smu_packing}.

As described above, the Necklace Splitting is, from the complexity point of view, hard already for the number of thieves $k=2$, and each $a_i= 1$.
Clearly, when the number of thieves $k=2$ and each $a_i= 1$ then the proof of Theorem \ref{thm:splitting} works for any set of numbers $M_1, \ldots, M_s$ with the property that all partial sums of $M_i$'s are pairwise distinct. 
 A natural question is whether such set exists with all $M_i$ bounded by a fixed power of $s$. 

It turns out that the answer is negative. This is related to a very nice part of the combinatorial number theory which we now explain.

\begin{definition}(Set with distinct subset sums) 
A set $S$ of positive integers has \emph{distinct subset sums}\index{distinct subset sum} if the set $\{\sum_{x\in X} x: X \subset S\}$ has $2^{|S|}$ distinct elements.
\end{definition}

For example, any set of distinct powers of number 2 has the distinct subset sums property. More examples of sets with distinct subset sums are $\{3,5,6,7\}$ and  $\{6,9,11,12,13\}.$ We mention a lower bound for the value of the maximum in the sets with the distinct subset sums property.

\begin{definition}
Let $f(n) = \min\{\max S: |S| = n, S$ has distinct subset sums $\}$.
\end{definition}

Paul Erd{\"o}s conjectured in 1931 that for some constant $c$
$$f(n) \geq c2^n.$$

Conway and Guy \cite{CG} found a construction of sets with distinct subset sum, now called the {\em Conway-Guy sequence}, which gives an upper bound on $f$. This was later improved by Lunnan \cite{Lu}, and then by Bohman \cite{Boh} to $f(n) \leq 0.22002 \cdot 2^n$ (for $n$ sufficiently large).

The best known lower bound, up to the constant, has been proved by Erd{\"o}s and Moser \cite{E2} in 1955, 
$$f(n) \geq 2^n/(10\sqrt n).$$

\subsection{Proof of Theorem \ref{thm:fix_degree_traverse}}
\label{sub.t3}

\begin{proof}
The length of the route has to be $l = 2 \sum_{e \in E} f(e)$ and let $I = \{ 1, \ldots, l \}$ be the set of all indices on the route. For every $A \subseteq I$ and $v \in V$ let $M_v[A]$ be true if there exists route satisfying all conditions on $T[v]$ using exactly indices of $A$ on $T[v]$. Similarly we define $M'_v[A]$ for $T'[v]$. Let $z(A)$ for the set of ordered pairs of starting and ending indices of subsequencies of $A$, i.e. $z(A) = \{(a_1,b_1), \ldots, (a_q,b_q)\}$ such that $A = \{a_1, \ldots, b_1\} \cup \cdots \cup \{a_q, \ldots, b_q\}$ and $a_1 \le b_1 < b_1+1 < a_2 \le b_2 < b_2+1 < \cdots < b_{q-1}+1 < a_q \le b_q$. Let $|z(a)| = q$ be the number of subsequences.

Let $v$ be a non-root vertex and $e = vp(v)$. If $M'_v[A] = true$, then $|z(A)| \le f(e)$ since $e$ has to be traversed $f(e)$ (some traverses may be consecutive). Therefore, there are at most $f(e) \cdot l^{2f(e)}$ sets $A$ such that $M'_v[A] = true$, so we can store all such sets $A$ instead of whole table $M'_v$ to ensure polynomial space. Similarly, if $M_v[A] = true$ then $|z(A)| \le f(e)$ since $T[v]$ can be entered at most $f(e)$-times.

We determine $M_v$ using the following dynamic programming. If $v$ is a leaf, then $M_v[A] = true$ only for $A = \emptyset$. Consider that $u_1, \ldots, u_s$ are all children of $v$. Recall that $1 \le s \le \Delta$. First, we set $M_v[A] := false$ for all $A$ and then we consider all combination $A_i$ for $i = 1, \ldots, s$ such that $M'_{u_i}[A_i] = true$. Note that there are at most $F^{\Delta} \cdot l^{2F\Delta} \le F^d \cdot (2Fn)^{2F\Delta}$ such combinations, so the algorithm is polynomial. Let $A = A_1 \cup \cdots \cup A_s$. We apply the following function for every combination.
\begin{itemize}
\item If any two sets of $A_1, \ldots, A_s$ have a common member, then the function terminates, since every index has to be used for exactly once on the route.
\item If $z(A) > f(e)$, then the function terminates, since $T[v]$ can be entered at most $f(e)$-times where $e = vp(v)$.
\item In the end, we set $M_v[A] := true$.
\end{itemize}

Now, we determine $M'_v$. Let $e= \{v,p(v)\}$. First, we set $M'_v[A] := false$ for all $A$ and then we apply the following function for every $A$ with $M_v[A] = true$.
\begin{itemize}
\item 
Let $z(A) = \{(a_1,b_1), \ldots, (a_q,b_q)\}$. If $q> f(e)$ then stop.
\item 
Let $X_1= \{a'_i; 1\leq i\leq f(e)\}$, $X_2= \{b'_i; 1\leq i\leq f(e)\}$ and $X= X_1\cup X_2$ be such that (1) $X\cap A= \emptyset$ 
(2) $X_1\cap X_2= \emptyset$ and 
(3) for each $(a_i, b_i)\in z(A)$, $a_i-1\in X_1$ and $b_i+1\in X_2$. For each such $X_1, X_2$ we let $A'= A\cup X$. 
\item
We check if $X_1, X_2$ satisfy the conditions for $g(e)$: if not, we stop.


\item 
We set $M'_v[A'] := true$.
\end{itemize}

Finally, the algorithm returns $M_d[I]$.
\end{proof}

We note that the same proof works if we require that every edge $e$ is traversed {\em at most} $f(e)$-times in both directions.

\begin{theorem}
\label{thm.main}
Let $z, \Delta, F$ be integer constants and let $(T,d,z,p)$ be a maintaining plan where $T$ is a tree with maximum degree $\Delta$. Let $f: E\rightarrow N$ satisfies for each $e\in E$, $f(e)\leq F$. Then there is a polynomial algorithm  to decide if a $L,c,t,f-$vehicle route on $T$ exists.
\end{theorem}
\begin{proof}
We use Theorem \ref{thm:fix_degree_traverse} and note that we can require that every edge $e$ is traversed {\em at most} $f(e)$-times in both directions, function $t$ can be modelled by $g$ and the capacity constraint can be modelled by connecting the depot to a new vertex of degree one and setting the proper value on $g(e)$ for the new edge. 

\end{proof}

\subsection{Proof of Theorem \ref{thm.np1}}
\label{sub.t4}
\begin{proof}
The 3-partition problems ask to decide whether a given integers $a_1, \ldots, a_{3n}$ can be split into $n$ groups with the same sum. The problem is strongly NP-complete even if it is restricted to integers strictly between $S/2$ and $S/4$ where $S$ is the target sum. Note that in this case, every group has to contains exactly 3 integers.

Let $h = \lceil \log_2 3n \rceil$ and $B = 3h$ and $b_i = B a_i$ for all $i = 1, \ldots, 3n$ and $S' = B(S+1)$.

Let $T_i$ be a binary tree on $b_i$ edges rooted in $r_i$. Let $T'$ be a binary tree rooted in $r'$ with leaves $r_1, \ldots, r_{3n}$ such that all leaves are in depth $h$. Let $T$ be a binary tree such that
\begin{itemize}
    \item $d$ is the root of $T$
    \item $d$ has only one child $d'$
    \item $T'$ is attached to the node $d'$
    \item trees $T_1, \ldots, T_{3n}$ are attached to leaves of $T'$.
\end{itemize}
Note that the size of $T$ is $O(\log n \sum_i a_i)$, so it is only $O(\log n)$-times larger than the size of the instance of 3-partition problem.

Next, $f(e) = 1$ for all edges $e$ on trees $T_1, \ldots, T_{3n}$. For an edge $e$ of $T'$, $f(e)$ is the number of trees of $T_1, \ldots, T_{3n}$ in the subtree of $e$. Finally, $f(dd') = n$. The goal is to ensure there that the route can be split into $n$ parts by passing $dd'$ and each part traverses from $d'$ to some $r_i$, whole tree $T_i$, returns to $d'$ and then traverse two other trees of $T_1, \ldots, T_{3n}$. In order to ensure the proper sum, we set $g(dd') = 2S'+2$ and $g(e)$ is a sufficiently larger number for all other edges $e$. Clearly, if integers can be split into $n$ groups, there exists a route.

Consider a walk $w$. Clearly, every tree $T_1, \ldots, T_{3n}$ has to be traversed by $w$ completely once it is entered. Traverses of $dd'$ split $w$ into $n$ parts and every tree $T_1, \ldots, T_{3n}$ is completely traversed in one part. Note that $g(dd')$ ensures that one part traverses at most $S'$ edges in both directions (excluding $dd'$).

We prove that every part traverse exactly tree trees. For contradiction, assume that trees $T_i, T_j, T_k, T_l$ are traversed in one part. Then, the number of edges in the part is at least $b_i + b_j + b_k + b_l + h = B (a_i + a_j + a_k + a_l) + h \ge 4B \frac{S+1}4 + h = B(S+1) + h > S'$ which is a contradiction. Since all parts contain at most $3$ trees, every part must contains exactly $3$ trees.

Now, consider a part traversing trees $T_i, T_j, T_k$. For contradiction, assume that $a_i + a_k + a_j > S$. The number of traversed edges in the part is at least $b_i + b_j + b_k + h = B(a_i + a_j + a_k) + h \ge B(S+1) + h > S'$ which is a contradiction. Hence, the sum of integers corresponding to each group is exactly $S$.

\end{proof}

\subsection{Proof of Theorem \ref{thm.np2}}
\label{sub.t5}
\begin{proof}
Consider an instance of 3-partition consisting of $3n$ integers $a_1, \ldots, a_{3n}$ and let $S$ be the target sum. Let $B = n$ and $b_i = B a_i$ for $i = 1, \ldots, 3n$ and $S' = BS+1$ and $S'' = (n-1)S'$. Let $T$ be tree which consists of
\begin{itemize}
\item a depo $d$, and
\item $n$ vertices $u_1, \ldots, u_n$ incident only to $d$ where $f(du_i) = 2$ and $g(du_i,1) = S'$ and $g(du_i,2)=S''$, and
\item $3n$ paths $P_1, \ldots, P_{3n}$ on $b_1, \ldots, b_{3n}$ edges such that one end-vertex of each path is $d$ and these paths have to be traversed only once.
\end{itemize}
Note that the only possible length of a route is $S'+S''$ and the short and long distances between tranverses of edges $du_i$ have to be exactly $S'$ and $S''$, respectively. If $a_1, \ldots, a_{3n}$ can be partitioned into $n$ groups of equal sum $S$ then we construct a route as follows: Starts by traversing $du_1$, tranverse 3 paths of the first group, tranverse $du_1$, traverse $du_2$, etc.

Observe that there is no route such that two close traverses of an edge $du_i$ which is interleaved by a traverse of an edge $du_j$ since the sum of lengths of any subset of paths is divisible by $B$ and even all edges edges $du_1,\ldots,du_n$ cannot contribute to a multiple of $B$. Hence, if there exists a route then it looks like as the one constructed above.
\end{proof}
Note that the problem is NP-complete even if $F = 2$ and all vertices except one have degree at most 2.

\subsection{Proof of Theorem \ref{thm.c}}
\label{sub.mmm}

Theorem \ref{thm.c} immediately follows from the following

\begin{theorem} \label{thm.tww}
Fixed integers $F, \Delta, k$. There exists a polynomial time algorithm which for a graph $G = (V,E)$ rooted in $d$ and with maximal degree at most $\Delta$, given along with its canonical tree decomposition $(W,b)$ of width $k-1$ and functions $f: E^s \to N$ such that $f(e) \le F$ for all $e \in E^s$ and $g: E^s\times \{1, \ldots, F\} \to N$ decides whether there exists a closed walk $w$ starting at $d$ satisfying
\begin{itemize}
\item Every edge $e$ is traversed $f(e)$-times in both directions.
\item For every edge $e$ and $y\leq f(e)$, there are at most $g(e, y)$ steps between the $y-$th and $(y+1)-$st traverses of $e$, taken cyclically.
\end{itemize}
\end{theorem}
\begin{proof}
We assume that for each bag $b(v)$, the edges of $G^s$ incident to a vertex of $b(v)$ (there are at most $2k\Delta$ of them) are linearly ordered. The ordering may differ in different bags.

The length of the route has to be $l = 2 \sum_{e \in E} f(e)$ and let $I = \{ 1, \ldots, l \}$ be the set of all indices on the route. 
Let $I'= \{(x,i); x\in I, 0\leq i\leq 2k\Delta\}$. For every $A \subseteq I'$,  and $u\in V(W)$ let $M_u[A]$ be true if there exists route $w= (e_1, \ldots, e_l)$ satisfying all conditions on $G_{p(u),u}$ so that: 
\begin{itemize}
    \item If $A_0= \{x;$ there is $i$ such that $(x,i)\in A\}$ then $w$ uses exactly indices of $A_0$ on $G_{p(u), u}$,  
    \item 
    For each $(x,i)\in A$, $i=0$ iff $e_x$ is not incident to a vertex of $b(u)$. 
    \item
    Let $S(A)= \{x\in A_0; e_x$ is incident with a vertex of $b(u)\}$.
    For each $x\in S(A)$, if $(x,i)\in A$ then the edge $e_x$ of $w$ is the $i-$th edge of the fixed linear order of the edgers incident with a vertex of $b(w)$.  
\end{itemize}


Let $z(A) = \{(a_1,b_1), \ldots, (a_q,b_q)\}$ such that $A_0 = \{a_1, \ldots, b_1\} \cup \cdots \cup \{a_q, \ldots, b_q\}$ and $a_1 \le b_1 < b_1+1 < a_2 \le b_2 < b_2+1 < \cdots < b_{q-1}+1 < a_q \le b_q$. Clearly, $|z(a)| = q$ be the number of subsequences in $A_0$.

If $M_u[A] = true$ then $|z(A)| \le |S(A)| \le 2Fk\Delta$ since $G_{p(u),u}$ can only be entered from a vertex of $b(u)$ which is incident with at most $\Delta$ edges in $G_{p(u),u}$ and each such edge can be used at most $f(e)-$times. Therefore, there are at most $(2kl\Delta)^{2Fk\Delta}$ sets $A$ such that $M_u[A] = true$, so we can store all such sets $A$ instead of whole table $M_u$ to ensure polynomial space.

We determine $M_u$ using the following dynamic programming. Let $u$ be a non-root vertex of $W$. 

If $u$ is a leaf, then $M_u[A] = true$ only for $A = \emptyset$. 

If $u$ is unique son of $p(u)$ and $b(p(u)) = b(u) \setminus \{v\}$ for some vertex $v$ of $G$ then $M_{p(u)}[A] = true$ iff $M_u[A'] = true$ where $A'$ obtained from $A$ but correcting the contribution of the linear order of edges associated with $b(p(u))$.

Let $u$ be the unique son of $p(u)$ and $b(p(u)) = b(u) \cup\{v\}$ for some vertex $v$ of $G$. We notice that no edge incident with $v$ belongs to $G_{p(u),u}$.
We construct sets $A$ for which $M_{p(u)}[A] = true$ 
by considering edges from $v$ to $b(u)$ one by one and for each such edge $e$ we perform the same construction as the one of $M'_v[A']$ from $M_v[A]$
in the proof of Theorem \ref{thm:fix_degree_traverse}.
     
Finally let $p(u)$ have two sons $u=u_1,u_2$. We know $b(p(u)) = b(u_1)= b(u_2)$. Let $S= E^s(G_{p(p(u_1),u_1}\cap E^s(G_{p(p(u_2),u_2})$. We observe: if $e\in S$ then $e$ is incident with a vertex of $b(p(u_1))$. 

We let again $M_{p(u)}[A] = false$ for each $A$ and do the following:

Consider all pairs $A_1, A_2$ such that $M_{u_1}[A_1] = true$ and $M_{u_2}[A_2] = true$. We first modify the elements of both $A_1, A_2$ to reflect the fixed linear order of the edges incident with a vertex of $b(p(v)$; this linear order may be different from the linear order (of thesame set) fixed for
$b(u_1)$ or for $b(u_2)$. Let the resulting sets be denoted by $A'_1, A'_2$. 

For $e\in S$ let $i(e)$ denote its index in the fixed linear order of the edges incident with a vertex in $b(p(u_1))$. For each such $i(e)$ let
$S_1= \{x; (x,i(e))\in A'_1\}$ and $S_2= \{x; (x,i(e))\in A'_2\}$. If $S_1\neq S_2$ then stop.

If $[A'_1]_0\cap [A'_2]_0$ contain any other element then stop.

If $A'= A'_1\cup A'_2$ does not satisfy the requirements given by function $g$ on the edges incident with a vertex of $b(p(u))$ then stop.

Let $M_{p(u)}[A'] = true$.

Finally, the algorithm returns $M_d[I]$.

\end{proof}



\end{document}